\documentclass{article}
\usepackage{spconf,amsmath,graphicx}

\usepackage[T1]{fontenc}
\usepackage{fancybox}
\usepackage{calc}
\usepackage{amsthm}
\usepackage{amssymb}
\usepackage{verbatim}
\usepackage{amsmath}
\usepackage{tabularx}

\newtheorem{theorem}{Theorem}
\DeclareMathOperator*{\argmin}{arg\,min}

\title{Data Discovery Using Lossless Compression-Based Sparse Representation}
\name{Elyas Sabeti$^{\star}$, Peter X.K. Song$^{\star \dagger}$ and Alfred O. Hero III$^{\dagger}$\thanks{This work was supported in part by Michigan Institute for Data Science and the Prometheus  Program  of the  Defense  Advanced  Research  Projects  Agency  (DARPA),  grant  number N66001-17-2-401.}}
\address{$^{\star}$Michigan Institute for Data Science, University of Michigan, Ann Arbor\\
      $^{\star \dagger}$Department of Biostatistics, School of Public Health, University of Michigan, Ann Arbor\\
			$^{\dagger}$Department of Electrical Engineering and Computer Science, University of Michigan, Ann Arbor}
\begin{document}
%\ninept
%
\maketitle
\begin{abstract}
Sparse representation has been widely used in data compression, signal and image denoising, dimensionality reduction and computer vision. While overcomplete dictionaries are required for sparse representation of multidimensional data, orthogonal bases represent one-dimensional data well. In this paper, we propose a data-driven sparse representation using orthonormal bases under the lossless compression constraint. We show that imposing such constraint under the Minimum Description Length (MDL) principle leads to a unique and optimal sparse representation for one-dimensional data, which results in discriminative features useful for data discovery. 

\end{abstract}
\vspace{-0.15in}
\section{Introduction}
\label{sec:intro}
\vspace{-0.1in}
Sparse representation has been historically used for data compression, signal and image denoising, and dimensionality reduction. Initially (before 1993), the foundation of sparse representation was based on certain transforms (e.g. Fourier, wavelets) and orthogonal bases that uses a fixed dictionary \cite{rubinstein2010dictionaries,zhang2015survey}. While these methods served one-dimensional data well, the dictated \textit{orthogonality} requirement was overly restrictive for multidimensional data \cite{rubinstein2010dictionaries}. Following the works of Mallat et al. \cite{mallat1993matching} in 1993 and Chen et al. \cite{chen1994basis} subsequently, a new era of sparse representation begin: instead of a fixed orthogonal dictionary, adaptive and \textit{overcomplete} dictionaries were used for sparse representation \cite{rubinstein2010dictionaries}. The overcomplete dictionaries permit multiple sparse descriptions in the representation domain, in which the best description is task-dependent. In recent years, data-driven dictionaries (dictionary learning) \cite{tosic2011dictionary} and convolutional sparse coding methods \cite{bristow2013fast,wohlberg2014efficient} have been used for variety of machine learning applications such as computer vision \cite{wright2008robust}.

In this paper, we develop a sparse representation using orthonormal (e.g. wavelets) bases under the lossless compression constraint, we show that imposing such constraint leads to a unique and optimal sparse representation for a one-dimensional data which can be used for classification and anomaly detection in time series. This paper is organized as follows: in Section \ref{sec:SparRep}, first the model is described, lossless compression overview is provided and then the solution to the problem is provided. Section \ref{sec:Experiment} presents an experiment that shows the power of the proposed optimal sparse representation in data pattern discovery.

\vspace{-0.17in}
\section{Sparse Representation Using Orthonormal Bases}
\label{sec:SparRep}
\vspace{-0.1in}
Let $\mathbf{W}\in\mathbb{R}^{l\times l}$ be an orthonormal Discrete Wavelet Transform (DWT) matrix whose column vectors are the basis elements of an orthonormal basis $\mathcal{B}$. For any data block $\mathbf{x}\in\mathbb{R}^{l}$ of length $l=2^N$ we have $\boldsymbol{\alpha} =\mathbf{W} \mathbf{x}$ where $\boldsymbol{\alpha}\in\mathbb{R}^{l}$ is a column vector containing the wavelet coefficients of $\mathbf{x}$. This leads to the following sparse representation
\vspace{-0.04in}
\begin{align}
\mathbf{x} & =\mathbf{W}^T\boldsymbol{\alpha}^{(k)}+\mathbf{n}, \qquad \|\boldsymbol{\alpha}^{(k)}\|_0 = k.
\label{eq:MainModel}
\end{align}
\vspace{-0.03in}
where $\boldsymbol{\alpha}^{(k)}$ is the vector of the wavelet coefficients with only $k$ non-zero coefficients and $\mathbf{n}\sim\mathcal{N}(0,\sigma^{2}\mathbf{I})$ is Gaussian white noise with unknown $\sigma^{2}$. The objective is to choose the value of $k$ to minimize an error criterion. Without any further constraints, there is no unique value for $k$, all of which result in a sparse model with loss of information (with the obvious exception of $k=l$). A similar problem has been considered in \cite{saito1994simultaneous} where the ``best'' $k$ and the best basis over a library of orthonormal bases was found in order to suppress the noise. Here, we show that using the lossless compression as the constraint and adhering to strict decodability \cite{sabeti2019data,host2019data} can lead to an ``optimal'' $k$ (and corresponding wavelet coefficients) whose variation in time series can be used as a discriminating feature in machine learning and anomaly detection setting. The process of finding such $k$ leads to Rissanen's famous Minimum Description Length (MDL) approach \cite{Rissanen83} that provides a framework tailored for optimization and model selection in a lossless compression setting. The other prominent penalty-based model selection methods that can be used are the Akaike Information Criterion (AIC) \cite{akaike1974AIC} and the Bayesian Information Criterion (BIC) \cite{schwarz1978BIC}, but it has been shown that MDL significantly outperforms AIC and BIC for generalized linear models \cite{hansen2003minimum,chen2005comparative}. MDL has also been previously used for sparse lossless audio compression \cite{ghido2012sparse} and dictionary learning \cite{ramirez2012mdl} as well. In the next section, we summarize the MDL results that are used in this paper. 

\subsection{Descriptive Length for Integers, Real-valued Numbers and Parametrized Models}
\label{MDLreal.sec}
The MDL principle attempts to minimize the number of bits required to parsimoniously describe (encode) observation. In this framework, the encoding process is not of interest and only the corresponding codelength is considered. MDL is often used in model selection frameworks in which the optimum model is the one that minimizes the summation of model complexity codelength and the negative log-likelihood. This is similar to how Structural Risk Minimization is used to balance the bias-variance trade-off in machine learning \cite{shalev2014understanding}. Suppose we want to describe the data with a parametrized model $f(\mathbf{x}|\boldsymbol{\theta})$. Let 
\begin{align*}
\mathcal{M}=\left\{ \boldsymbol{\theta}_{m}\in\Theta_{m}\subset\mathbb{R}^{k_{m}}:m=1,2,\ldots\right\},
\end{align*}
be a class of models at hand where the integer $m$ is an index of the model in the list. In order to encode a sequence $\mathbf{x}$ of length $l$ of the observed data using the model $\boldsymbol{\theta}_{m}$ without losing information (lossless), we have to encode the index $m$, the model $\boldsymbol{\theta}_{m}$ given $m$ and the data itself given the model, therefore the total codelength of the whole process is
\begin{align}
L(\mathbf{x},\boldsymbol{\theta}_{m},m) & =L(m)+L(\boldsymbol{\theta}_{m}|m)+L(\mathbf{x}|\boldsymbol{\theta}_{m},m).\label{eq:General MDL}
\end{align}
The idea of the MDL principle is to minimize the total codelength $L(\mathbf{x},\boldsymbol{\theta}_{m},m)$ while balancing the trade-off between the model complexity $L(\boldsymbol{\theta}_{m}|m)$ and data complexity $L(\mathbf{x}|\boldsymbol{\theta}_{m},m)$ in order to avoid over-fitting. As such, $L(\mathbf{x}|\boldsymbol{\theta}_{m},m)$ is minimized by using the maximum likelihood estimator $\boldsymbol{\hat{\theta}}_{m}$, thus $L(\boldsymbol{\hat{\theta}}_{m}|m)$ can be expressed using its $k_{m}$ real-valued parameters estimates $\boldsymbol{\hat{\theta}}_{m}$ either
by $\frac{5}{2}k_{m}+\frac{k_{m}}{2}\log l$ \cite{Host15ISITl}, or by a $\hat{\boldsymbol{\theta}}_{m}$-dependent expression $\overset{k_{m}}{\underset{j=1}{\sum}}\log^{*}\left(\left\lfloor \hat{\theta}_{m,j}\right\rfloor \right)+\frac{k_{m}}{2}\log l$ where $\log^{*}(l)=\log l+\log\log l+\log\log\log l+\cdots$, where
the sum continues as long as the argument to the log is positive,
both of which originated from Rissanen's prior for the integers \cite{Rissanen83}.
In summary we have
{\small{}
\begin{align}
L(m) & =\begin{cases}
\log M, & \qquad|\mathcal{M}|\leq M;\\
\log^{*}m, & \qquad\quad o.w.
\end{cases}\nonumber \\
L(\boldsymbol{\hat{\theta}}_{m}|m) & =\begin{cases}
\frac{5}{2}k_{m}+\frac{k_{m}}{2}\log l, & \!\!\!\!\!\!\!\!\!\!\!\!\!\!\!\!\!\! independent of\,\hat{\theta}_{m,j};\\
\overset{k_{m}}{\underset{j=1}{\sum}}\log^{*}\left(\left\lfloor \hat{\theta}_{m,j}\right\rfloor \right)+\frac{k_{m}}{2}\log l, & \qquad o.w.
\end{cases}\label{eq:expr for Gpar}\\
L(\mathbf{x}|\boldsymbol{\hat{\theta}}_{m},m) & =-\log f(\mathbf{x}|\hat{\boldsymbol{\theta}}_{m},m).\nonumber 
\end{align}
}
The minimum of (\ref{eq:General MDL}) using the expressions in
(\ref{eq:expr for Gpar}) gives the best trade-off between the
model complexity and the likelihood of the data without any loss of information. 

\subsection{The Optimal Sparse Representation Using MDL}
\label{sec:SparRepMDL}
Going back to the problem considered earlier in the section, we aim to find the optimum sparse representation in equation (\ref{eq:MainModel}) subject to achieving the minimal lossless codelength of the observed data under the MDL principle. The following theorem presents such optimal $k$ and gives the corresponding minimum descriptive codelength for this problem:
\begin{theorem}
\label{thm:OptimalK}
Suppose $\mathbf{x}\in\mathbb{R}^{l}$ is observed data of length $l=2^N$ contaminated with an additive Gaussian noise $\mathbf{n}$ with unknown variance $\sigma^{2}$, let $\mathbf{W}\in\mathbb{R}^{l\times l}$ be an orthonormal matrix. Given the sparse model $\mathbf{x} =\mathbf{W}^T\boldsymbol{\alpha}^{(k)}+\mathbf{n}$ such that $\|\boldsymbol{\alpha}^{(k)}\|_0 = k$, subject to the minimization of the lossless descriptive length of $\mathbf{x}$, the optimal value of $k$ and the codelength of $\mathbf{x}$ are given by
{\small{}
\begin{align*}
k & = \argmin_{1\leq k<l/2}\left\{ \frac{5}{2}k+\frac{k}{2}\log l+lH(\frac{k}{l})+\frac{l}{2}\log\left\Vert \nabla^{(l-k)}\mathbf{W}\mathbf{x}\right\Vert ^{2}\right\},\\
L & =\min_{1\leq k<l/2}\left\{ \frac{5}{2}k+\frac{k}{2}\log l+lH(\frac{k}{l})+\frac{l}{2}\log\left\Vert \nabla^{(l-k)}\mathbf{W}\mathbf{x}\right\Vert ^{2}\right\} \\
 & +\frac{5}{2}+2\log l+\frac{l}{2}\log\frac{2\pi}{l}+\frac{l}{2\ln2}+\log\frac{\pi}{8},
\end{align*}}
where $H(.)$ is the entropy function and $\nabla^{(k)}\mathbf{y}$ is an operator that keeps the $k$ components of $\mathbf{y}$ with the smallest absolute value.
\end{theorem}
\begin{proof}
According to equation (\ref{eq:General MDL}), we first need to take into account the codelength required to describe (encode) the unknown value $k$; since $k\leq l$, by expression in (\ref{eq:expr for Gpar}), it requires $\log l$ bits. The next step is to find the maximum likelihood estimator of the $k+1$ parameters, i.e. $\boldsymbol{\alpha}^{(k)}$ and $\sigma^{2}$. It's easy to verify that the maximum likelihood estimation of the noise variance is $\widehat{\sigma}^{2}=\frac{1}{l}\left\Vert \mathbf{x}-\mathbf{W}^T\boldsymbol{\alpha}^{(k)}\right\Vert ^{2}$
and the maximum likelihood estimation of $k$ non-zero coefficients of $\boldsymbol{\alpha}^{(k)}$
is the largest $k$ coefficients of $\mathbf{W}\mathbf{x}$. Now by defining the operators $\triangle^{(k)}$ and $\nabla^{(k)}$ that keep the largest magnitude $k$ components and smallest magnitude $k$ components respectively and set rest to zero, we obtain
\begin{align*}
\boldsymbol{\widehat{\alpha}}^{(k)} & =\triangle^{(k)}(\mathbf{W}\mathbf{x}),\\
\widehat{\sigma}^{2} & =\frac{1}{l}\left\Vert \mathbf{x}-\mathbf{W}^T\boldsymbol{\alpha}^{(k)}\right\Vert ^{2}
=\frac{1}{l}\left\Vert \nabla^{(l-k)}\mathbf{W}\mathbf{x}\right\Vert ^{2}.
\end{align*}
In addition to the values (maximum likelihood estimates) of the $k+1$ parameters that we encode
using $\frac{5}{2}(k+1)+\frac{k+1}{2}\log l$, we have to convey the
index of the sequence with $k$ non-zero elements of $\boldsymbol{\alpha}^{(k)}$
among all the sequences with $k$ non-zero elements. This operation requires
$\log\binom{l}{k}$. Now using the upper bound from \cite[(13.46)]{CoverBook}
we have $\log\binom{l}{k}\leq lH(\frac{k}{l})+\frac{1}{2}\log l+\log\frac{\pi}{8}$ where $H(.)$ is the entropy function. Therefore,
{\small{}
\begin{align*}
L(k) & =\log l,\\
L(\boldsymbol{\widehat{\alpha}}^{(k)},\widehat{\sigma}^{2}|k) & =\frac{5}{2}\left(k+1\right)+\frac{k+1}{2}\log l+\log\binom{l}{k}\\
 & \approx\frac{5}{2}\left(k+1\right)+\left(\frac{k}{2}+1\right)\log l\\
 & +lH(\frac{k}{l})+\log\frac{\pi}{8},\\
L(\mathbf{x}|\boldsymbol{\widehat{\alpha}}^{(k)},\widehat{\sigma}^{2},k) & =\frac{l}{2}\log\left(\frac{2\pi}{l}\left\Vert \nabla^{(l-k)}\mathbf{W}\mathbf{x}\right\Vert ^{2}\right)+\frac{l}{2\ln2}.
\end{align*}
}
Ergo the optimal value of $k$ and the codelength of $\mathbf{x}$ are
{\small{}
\begin{align*}
k & = \argmin_{1\leq k<l/2}\left\{ \frac{5}{2}k+\frac{k}{2}\log l+lH(\frac{k}{l})+\frac{l}{2}\log\left\Vert \nabla^{(l-k)}\mathbf{W}\mathbf{x}\right\Vert ^{2}\right\},\\
L & =\min_{1\leq k<l/2}\left\{ \frac{5}{2}k+\frac{k}{2}\log l+lH(\frac{k}{l})+\frac{l}{2}\log\left\Vert \nabla^{(l-k)}\mathbf{W}\mathbf{x}\right\Vert ^{2}\right\} \\
 & +\frac{5}{2}+2\log l+\frac{l}{2}\log\frac{2\pi}{l}+\frac{l}{2\ln2}+\log\frac{\pi}{8}.
\end{align*}}
\end{proof}

\vspace{-.2in}
It's worth pointing out the two main conclusions from the theorem: (i) the MDL principle balances the trade-off between the number of retained coefficients and the error, the larger the $k$ the more bits required to encode $\frac{5}{2}k+\frac{k}{2}\log l+lH(\frac{k}{l})$ but the fewer bits required to encode $\log\left\Vert \nabla^{(l-k)}\mathbf{W}\mathbf{x}\right\Vert ^{2}$; (ii) the range from which $k$ is chosen is $[1,l/2]$ since the wavelet coefficients in the range $[l/2+1,l]$ are associated with the details at the lowest scale. This lowest scale usually contains the majority of the noise power, so the fact that as $k\rightarrow l$, we have $\log\left\Vert \nabla^{(l-k)}\mathbf{W}\mathbf{x}\right\Vert ^{2} \rightarrow -\infty$, which leads to a trivial solution \cite{mallat2008wavelet,percival2000wavelet}. Figure \ref{fig:ComplexityVsError} illustrates the aforementioned trade-off for an autoregressive process.

\vspace{-.15in}
\begin{figure}[htb]
\centering
  \centerline{\includegraphics[width=8.5cm]{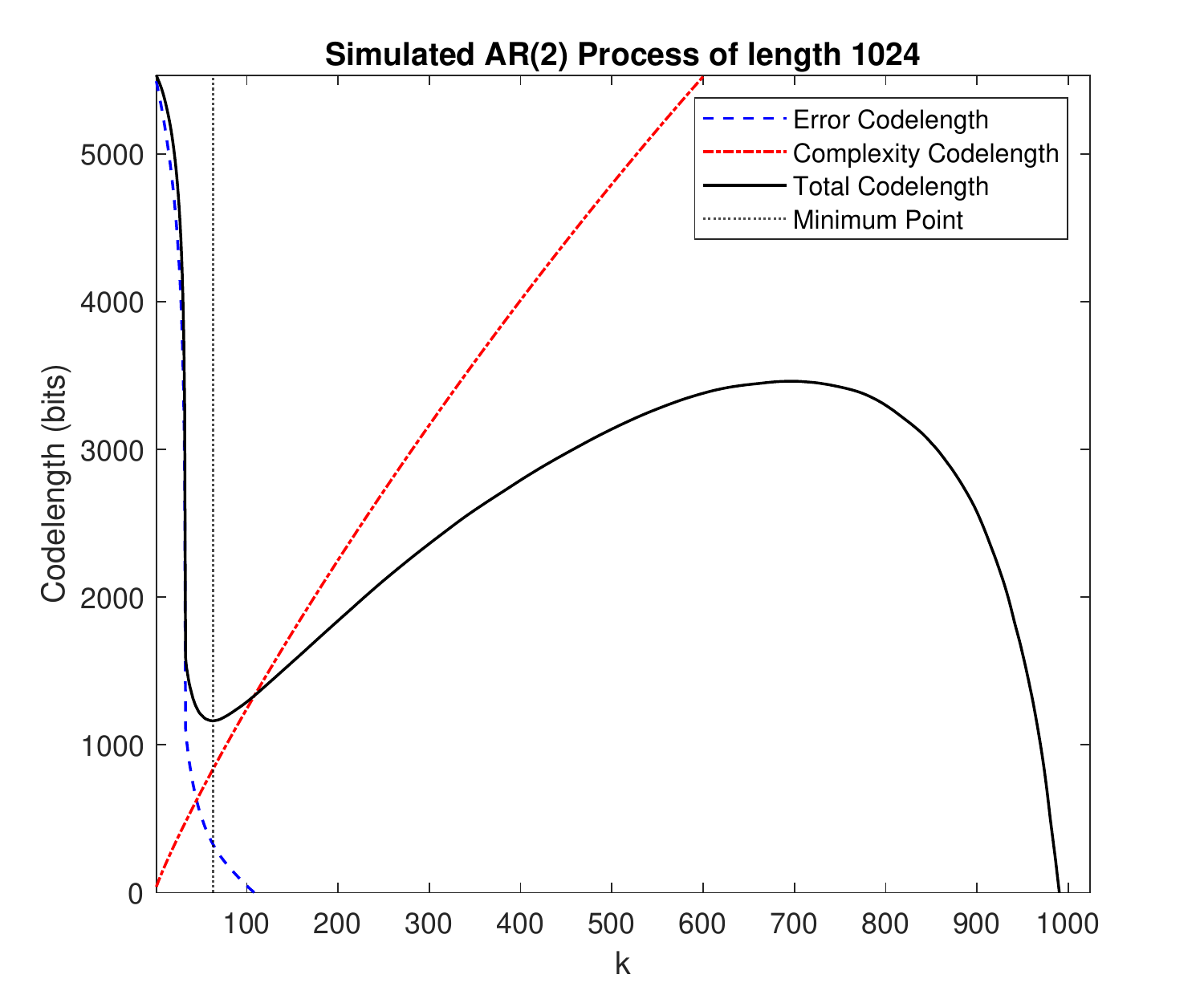}}
  \vspace{-.25in}
  \caption{\label{fig:ComplexityVsError} Error codelength, complexity codelength, the total codelength and the optimum $k$ for an AR(2) process.}\medskip
\end{figure}

\vspace{-.35in}
\section{Experiment}
\label{sec:Experiment}
\vspace{-.15in}
In order to show the discriminative power of the features extracted from the proposed sparse representation given in Theorem \ref{thm:OptimalK}, we use an experimental dataset in which a perturbation occur after baseline. The dataset consists of signals (temperature, heart rate, accelerometer and electrodermal activity) recorded by a wearable device (Empatica E4) collected as part of a human viral challenge study. In this study, the data from participants were collected for three consecutive days before and five consecutive days after their exposure to Human Rhinovirus (HRV) pathogen in an isolated quarantine ward. During these eight days, the wearable time series were continuously recorded while the biospecimen collection (gene expression, metabolomics, viral load) took place on a daily basis. Based on the measured viral load, subjects are divided into shedders (infected) and non-shedders. While the infection status is clearly detectable from biospecimen collections, we would like to analyze the time series from the non-invasive wearable device to potentially detect infection. In this paper we mainly focus on the temperature time series due to likely fever caused by infectious diseases.

In the preprocessing phase, we first downsampled the time series to one sample per 10 seconds and then removed any outliers, e.g. due to loss of contact, which is easy to detect due to the sudden drop in the recorded temperature. In order to choose a wavelet basis, we used all the orthonormal wavelet libraries in the Wavelab's MATLAB package \cite{buckheit1995wavelab,buckheit1995wavelab_new} with various number of parameters to calculate the total codelength of the time series using Theorem \ref{thm:OptimalK}, which were used to choose the wavelet basis that resulted in the shortest codelength. While the Daubechies wavelets achieved the shortest codelength, Daubechies 8 was better than the rest. After choosing the basis, we used a sliding window of length 256 and calculated the optimal value of $k$ (Theorem \ref{thm:OptimalK}) for each segment. Figure \ref{fig:ShedderNonShedder} depicts the mean and the standard deviation of the optimal value of $k$ (represented as percentage, i.e. $\frac{k}{256}$) for all days using the error bar representation. As seen, there is a significant difference before and after inoculation (exposure to virus pathogen) for the shedders while such changes do not occur for the non-shedders. In order to show the importance of the lossless constraint and the strength of MDL approach, we compare the discriminating property of the optimal $k$ of Theorem \ref{thm:OptimalK} with the $k$'s achieved using BIC and AIC. This result is summarized in Table \ref{tab:Comparison} showing that the changes in the time series before and after inoculation are not reflected in the optimal $k$'s of BIC and AIC, while the optimal $k$ of Theorem \ref{thm:OptimalK} is sensitive to these changes. Other types of discriminative features that can be extracted are statistics of the selected wavelet coefficients. How these features could be used for prediction, is an ongoing work and is not described here. Instead, in the next section we show how the proposed sparse representation and the resulting discriminating features can be used in atypicality for anomaly detection.

\begin{figure}[htb]
\centering
  \centerline{\includegraphics[width=8.5cm]{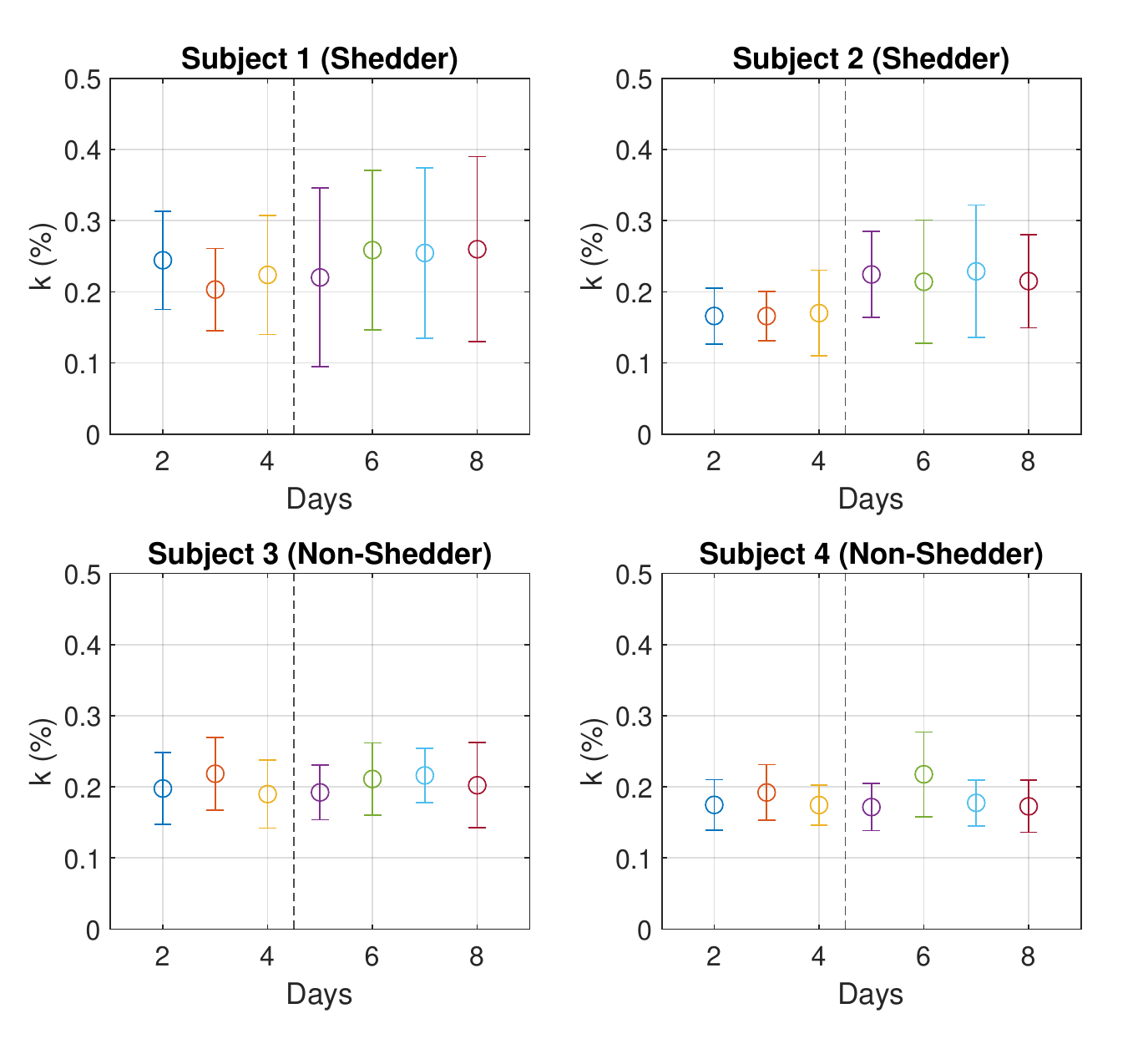}}
  \vspace{-.25in}
  \caption{\label{fig:ShedderNonShedder} The error bar representation of the daily optimal value of $k$ (represented as percentage, i.e. $\frac{k}{256}$) for two shedders and two non-shedders. The vertical dashed line shows the inoculation time. The first day was excluded since the data was recorder for a shorter period of time.}\medskip
\end{figure}

\vspace{-.2in}
\subsection{Anomaly Detection Using the Optimal Sparse Representation with Atypicality}
\label{sec:atypicality}
\vspace{-.1in}
Atypicality is a data discovery and anomaly detection framework that is based on a central definition: ``a sequence is atypical if it can be described (coded) with fewer bits in itself rather than using the (optimum) code for typical sequences'' \cite{sabeti2019data,host2019data}. In the atypicality framework, the comparison of the descriptive codelength between a training-based typical encoder and a universal encoder (independent of the train data and any prior information) is the criterion for detecting anomalous segments of data, i.e. the \textit{atypical} subsequences. Since the atypicality framework adheres to strict decodability, it is a well-suited framework for the optimal sparse representation proposed in this paper in order to detect anomalies. 

In our experiment, the temperature times series of the first two days was used as the training data, and the rest was used as the test data. For the typical encoder (training-based encoder) since both encoder and decoder have access to the training data and the wavelet basis, a sliding window of length $l$ is used in order to create a \textit{dictionary} $\mathcal{D}$ of the optimal $k_t$'s and the corresponding indexes of the $k_t$ non-zero elements of $\boldsymbol{\alpha}^{(k)}$ for all the segments in the training data. Therefore similar to the training data, both encoder and decoder also have access to the dictionary $\mathcal{D}$ as well. Then for any length-$l$ segment $\mathbf{x}$ of the test data, the typical encoder finds the best set of parameters ($k_t$ and the indexes of the $k_t$ non-zero elements) from the dictionary and sends its dictionary index to the decoder using $\log |\mathcal{D}|$ bits along with the error and the parameters values. As such, the typical codelength is
\begin{align*}
L_{t} & =\min_{d\in \mathcal{D}}\left\{ \frac{5}{2}k_d+\frac{k_d}{2}\log l+\frac{l}{2}\log\left\Vert \nabla^{(l-k_d)}\mathbf{W}\mathbf{x}\right\Vert ^{2}\right\} \\
 & + \log |\mathcal{D}|+\frac{l}{2}\log\frac{2\pi}{l}+\frac{l}{2\ln2}.
\end{align*}
For the atypical encoder, we use the result of Theorem \ref{thm:OptimalK}, ergo the atypical codelength is
{\small{}
\begin{align*}
L_{a} & =\min_{1\leq k<l/2}\left\{ \frac{5}{2}k+\frac{k}{2}\log l+lH(\frac{k}{l})+\frac{l}{2}\log\left\Vert \nabla^{(l-k)}\mathbf{W}\mathbf{x}\right\Vert ^{2}\right\} \\
 & +\frac{5}{2}+3\log l+\frac{l}{2}\log\frac{2\pi}{l}+\frac{l}{2\ln2}+\log\frac{\pi}{8}+\tau,
\end{align*}
}
in which $\log l+\tau$ is also added as the penalty for not knowing the start and end points of the anomalous sequence in advance \cite{sabeti2019data,host2019data}, and then $\tau$ can be used as a detection hyperparameter for which a value can be derived by cross-validation. As such, let $L'_{a}=L_{a}-\tau$. Therefore, the detection criterion is $L'_a-L_t>\tau$. Figure \ref{fig:atypicality} depicts the temperature time series for a shedder and the detected anomalous segments, all of which occur after inoculation. Applying the same approach to non-shedders we did not find any anomalies.

\begin{table}[tbh]
\begin{centering}
\begin{tabular}{ | c | c | c | c | }
\hline
           & AIC & BIC & Theorem \ref{thm:OptimalK} \\ \hline
Subject 1 (Shedder)  & 1.05   & 1.15   & 1.73   \\ \hline
Subject 2 (Shedder)  & 0.99   & 0.79   & 1.70   \\ \hline
Subject 3 (Non-Shedder) & 0.97   & 0.76   & 0.93   \\ \hline
Subject 4 (Non-Shedder) & 1.02   & 1.07   & 1.17   \\ \hline
\end{tabular}
\par\end{centering}
\vspace{-.1in}
\caption{\label{tab:Comparison} Comparison between the ratio of after to before inoculation of the averaged standard deviation of the optimal value of $k$ achieved from AIC, BIC and Theorem \ref{thm:OptimalK} for the subjects in Figure \ref{fig:ShedderNonShedder}. Note that only for the optimal $k$ of Theorem \ref{thm:OptimalK}, this ratio is considerably greater than one for shedders and approximately one for non-shedders, which can be used for detection and classification purposes.}
\end{table}

\vspace{-.3in}
\begin{figure}[htb]
\centering
  \centerline{\includegraphics[width=8.5cm]{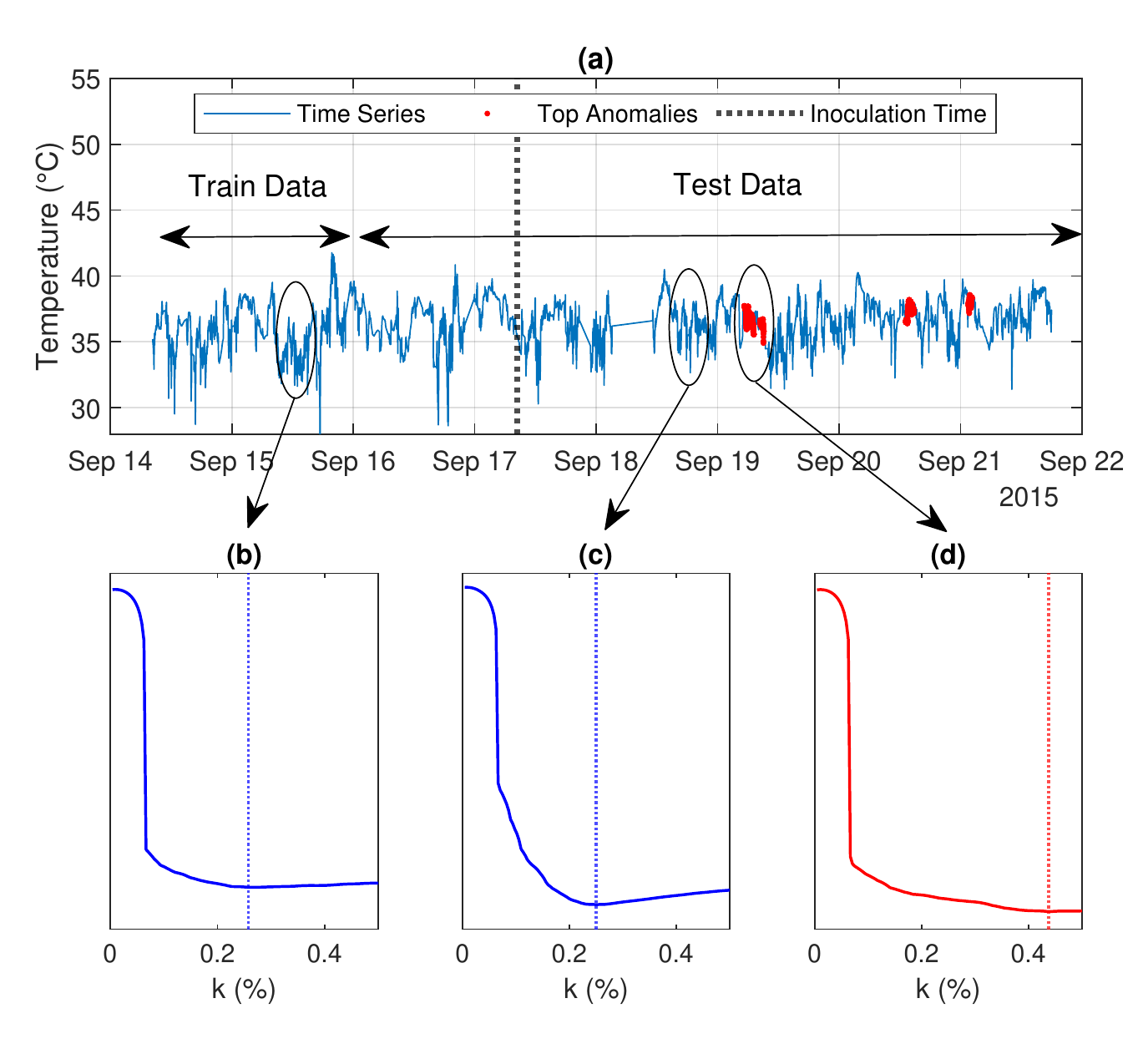}}
  \vspace{-.25in}
  \caption{\label{fig:atypicality} The result of anomaly detection for a subject that sheds virus: (a) the temperature times series and the detected anomalous segments (training and test data are shown by arrows, and inoculation time is illustrated with a vertical dashed line); the atypical codelength and the optimum $k$ for: (b) a random segment of train data; (c) a random non-anomalous segment of test data; (d) an anomalous segment of test data.}\medskip
\end{figure}

\vspace{-.3in}
\section{Conclusion}
\label{sec:conclusion}
\vspace{-.1in}
In this paper, we proposed a sparse representation under lossless compression constraint using orthonormal bases. We then used the MDL principle to achieve a unique and optimal sparse representation for one-dimensional time series data with application in data pattern discovery. In future works, we will use a comprehensive set of discriminative features extracted from the proposed data-driven sparse representation to enhance machine learning in discriminative analysis.

\bibliographystyle{IEEEbib}
\bibliography{Prometheus,ECGandHRV,BigData,Coop03}

\end{document}